\documentclass[12pt]{article}

\usepackage[utf8]{inputenc}

\usepackage{amssymb}
\usepackage{amsmath}
\usepackage{amscd}
\usepackage{amsthm}
\usepackage{amsfonts}
\usepackage{enumerate}
\usepackage{graphicx}
\usepackage[margin=1.0in]{geometry}
\usepackage{url}
\usepackage[breaklinks=true]{hyperref}
\usepackage[dvips]{color}
\usepackage{epsfig}

\newtheorem{theorem}{Theorem}[section]
\newtheorem{lemma}[theorem]{Lemma}
\newtheorem{proposition}[theorem]{Proposition}

\newtheorem{corollary}[theorem]{Corollary}

\theoremstyle{definition}

\newtheorem{problem}[theorem]{Open Problem}
\newtheorem{definition}[theorem]{Definition}

\newtheorem{remark}[theorem]{Remark}
\numberwithin{equation}{section}

\date{}
\author{Mark Braverman\thanks{Department of Computer Science, Princeton University, email:
mbraverm@cs.princeton.edu. Research supported in part by an NSF
CAREER award (CCF-1149888), a
Turing Centenary Fellowship, a Packard Fellowship in Science and
Engineering, and the Simons Collaboration on Algorithms and Geometry.} \and Jon Schneider\thanks{Department of Computer Science, Princeton University, email:
js44@cs.princeton.edu.}}
\title{Information complexity is computable}
\begin{document}

\maketitle

\begin{abstract}
The information complexity of a function $f$ is the minimum amount of information Alice and Bob need to exchange to compute the function $f$. In this paper we provide an algorithm for approximating the information complexity of an arbitrary function $f$ to within any additive error  $\alpha>0$, thus resolving an open question as to whether information complexity is computable.

 In the process, we give the first explicit upper bound on the rate of convergence of the information complexity of $f$ when restricted to $b$-bit 
protocols to the (unrestricted) information complexity of $f$. 
\end{abstract}

\section{Introduction}

In 1948, Shannon introduced the field of information theory as a set of tools for understanding the limits of one-way communication \cite{Sha48}. One of these tools, the information entropy function $H(X)$, measures the amount of information contained in a random source $X$. 

The analogue of information entropy in communication complexity is \textit{information complexity}. The information complexity of a function $f$ is the least amount of information Alice and Bob need to exchange about their inputs to compute a function $f$. Just as the information entropy of a random source $X$ provides a lower bound on the amount of communication required to transmit $X$, the information complexity of a function $f$ provides a lower bound on the communication complexity of $f$ \cite{BBCR10}. Moreover, just as Shannon's source coding theorem establishes $H(X)$ as the asymptotic communication-per-message required to send multiple independent copies of $X$, the information complexity of $f$ is the asymptotic communication-per-copy required to compute multiple copies of $f$ in parallel on independently distributed inputs \cite{BR11, Br12}.

These properties make information complexity a valuable tool for proving results in communication complexity.  Communication complexity lower bounds themselves have a wide variety of applications to other areas of computer science; for example, results in circuit complexity such as Karchmer-Wigderson games and ACC lower bounds rely on communication complexity lower bounds \cite{KW90, BT91}. In addition, techniques from information complexity have been applied to prove various direct sum results in communication complexity \cite{CSWY01, BJKS, Jai11}, including the only known direct sum results for general randomized communication complexity \cite{BBCR10}. Information complexity has also been applied to prove a tight asymptotic bound on the communication complexity of the set disjointness function \cite{BGPW}.

Despite this, many fundamental properties of information complexity remain unknown \cite{BGPW}. It is unknown how the information complexity of a function changes asymptotically as we allow the protocol to fail with probability $\epsilon$. It is unknown how the information complexity of a function grows if we restrict our attention to protocols of bounded depth. Perhaps most surprisingly, it is even unknown if, given the truth table of a function $f$, whether it is possible to even compute (to within some additive factor of $\epsilon$) the information complexity of $f$. (Contrast this with the case of the information entropy $H(X)$, which is easily computed given the distribution of $X$).

In this paper, we resolve the last of these questions; we prove that the information complexity of $f$ is indeed computable. Our main technical result is an explicit bound on the convergence rate of $r$-round information complexity to the unbounded-round information complexity.  More specifically, we show how to convert an arbitrary protocol $\pi$ into a protocol $\pi'$ that leaks at most $\epsilon$ more information than $\pi$, but requires at most $(N\epsilon^{-1})^{O(N)}$ rounds (Theorem \ref{convthm}). Equivalently, we show that the $r$-round information complexity of $f$ is at most $r^{-O(N^{-1})}$ larger than the information complexity of $f$. By combining this convergence results with prior results connection information and communication complexity, we obtain an algorithm that computes the information complexity of $f$ to within an additive factor of $\alpha$ in time $2^{\exp\left((N\alpha^{-1})^{O(N)}\right)}$ (here $N$ is the size of the truth table of $f$) . 

\subsection{Prior Work}

In \cite{MI08, MI09}, Ma and Ishwar present a method to compute tight bounds on the information complexity of functions for protocols restricted to $r$  rounds of computation. By examining the limit as $r$ tends to infinity, this method allows them to numerically compute the information complexity of several functions (such as the 2-bit $AND$ function). To make these computations provably correct, one would need effective (computable) estimates on the rate of convergence of $r$-round information complexity to the true information complexity. Such estimates were unknown prior to the present paper. 

Plenty of unsolved problems of this flavor --- where the computability of some limiting value is unknown despite it being straightforward to compute individual terms of this limit --- occur in information theoretic contexts. One famous problem is the problem of computing the Shannon capacity of a graph, the amortized independence number of the $k$th power of a graph (this limiting quantity also has an interpretation as the zero-error channel capacity of a certain channel defined by this graph). While computing the independence number of any given graph is possible (albeit NP-hard), the rate at which this limit converges is unknown. Indeed, Alon and Lubetzky have shown that the limiting behavior of this quantity can be quite complex; no fixed number of terms of this limit is guaranteed to give a subpolynomial approximation to the Shannon capacity \cite{AL06}. Another example, from the realm of quantum information theory, occurs in computing the quantum value of games \cite{CHTW04}. Here it is straightforward to compute the quantum value of a game when limited to $n$ bits of entanglement, but no explicit bounds are known for how many bits of entanglement are required to achieve within $\epsilon$ of optimal performance.

\subsection{Outline of Proof}

The main result of our paper is that zero-error information complexity is computable. Formally, we prove the following theorem.

\begin{theorem}\label{mainthm}
There exists an algorithm which, given a function $f: \mathcal{A}\times\mathcal{B} \rightarrow \{0, 1\}$, initial distribution $\mu \in \Delta(\mathcal{A}\times\mathcal{B})$, and a real number $\alpha > 0$, returns a value $C$ between $IC_{\mu}(f) - \alpha$ and $IC_{\mu}(f) + \alpha$. This can be performed in time  $2^{\exp\left((N\alpha^{-1})^{O(N)}\right)}$, where $N = |\mathcal{A}\times \mathcal{B}|$.
\end{theorem}

Throughout this paper, we will take the perspective of an outside observer watching in as Alice and Bob execute some protocol. This observer starts with some probabilistic \textit{belief} about the inputs of Alice and Bob (initially this is just $\mu$, the distribution of inputs to Alice and Bob). As Alice and Bob execute the protocol, they send each other \textit{signals} --- Bernoulli random variables that contain information about their inputs --- which cause the observer to update his belief. The total amount of information leaked by the protocol to the participants can then be represented directly in terms of the final belief and initial belief (Lemma \ref{icleaves}). These notions are defined in more detail in Section \ref{preliminaries}.

The strategy of the proof is as follows. We start with a general protocol $\pi$ for solving $f$, and whose information cost is very close 
to the information complexity of $f$. Unfortunately, we do not know anything about $\pi$ besides the fact that it's a finite, discrete protocol 
that computes $f$ without error. Note that if we could restrict $\pi$ to a finite family of protocols (e.g. protocols that sent at most $b$ bits, for
an explicit bound $b=b(\alpha,N)$, then we could just brute force over all such $\pi$'s and compute the approximate information complexity of $f$. 
The proof shows that, indeed, there is always a protocol $\pi'$ that can be derived from $\pi$, and which belongs to such an explicit family. 
The proof proceeds in several steps. In each step, more structure is added to $\pi$ (structure that is then exploited by the following steps). The difficulty is, of course, ensuring at each step that $\pi$ can be replaced with a more structured 
protocol $\pi'$ while increasing its information cost by only, say, $\alpha/10$. Ultimately, we manage to turn $\pi$ into a protocol 
with $r$ back-and-forth rounds, where $r$ is an explicit function of $N$ and $\alpha$. Finally, it is shown that an $r$-rounds of interaction 
protocol can be replaced with a $b$-bit protocol where $b=b(\alpha,N,r)=b(\alpha,N)$ is an explicit function, while only increasing its information cost by a controlled amount, completing the proof.

The actual proof of Theorem \ref{mainthm} is roughly structured into three parts. In the first part, we begin by showing that we can `discretize' any protocol $\pi$; that is, we can simulate any protocol $\pi$ with a protocol $\pi'$ that only uses a bounded number of different types of signals, but that only reveals a marginal amount of additional information. This takes several steps. In Section \ref{sectmu}, we show that it suffices to only consider initial distributions $\mu$ with full support. In Section \ref{sectboundary}, we show that we can modify protocols so that they never use signals too close to the boundary of $\Delta(\mathcal{A}\times\mathcal{B})$. In Section \ref{sectsize}, we show that, by dividing up large signals into smaller parts, we can modify protocols so that all signals have roughly the same `size'. Finally, in Section \ref{sectnum}, we show that for protocols with the previous three properties, we can apply a `rounding scheme' to each signal in this protocol and end up with a bound on the number of different signals in our new protocol. 

In the second part, we show in Section \ref{sectrounds} that we can transform any suitably discrete protocol $\pi$ (i.e. one that uses an explicitly bounded number of distinct signals) into a protocol that uses few rounds. We achieve this via a bundling scheme; the main idea is that, where Alice would ordinarily send Bob one instance of a signal, she instead sends Bob several instances of this signal. Then, the next several times Alice would send that signal to Bob, Bob simply refers to the next unused copy sent by Alice, thus decreasing the number of rounds in the protocol.

Combining the above steps allows us to prove the following bound on the convergence rate of $r$-round information complexity.

\newtheorem*{convthm}{\textbf{\em{Theorem \ref{convthm}}}}
\begin{convthm} \em{
Let $\pi$ be a communication protocol with information cost $C$ that successfully computes function $f$ over inputs drawn from distribution $\mu$ over $\mathcal{A}\times\mathcal{B}$. Then there exists a protocol $\pi'$ with information cost at most $C+\epsilon$ that also successfully computes $f$ over inputs drawn from $\mu$, but that uses at most $w(f, \epsilon)$ alternations where

\begin{equation}
w(f, \epsilon) = (N\epsilon^{-1})^{O(N)}
\end{equation}

\noindent
where $N = |\mathcal{A}\times\mathcal{B}|$.}
\end{convthm}

Finally, in the third part of the proof in Section \ref{sectic}, we demonstrate how to approximate the bounded-round information complexity of a function by computing the communication complexity of several parallel copies of this function. We accomplish this by combining an existing result of Braverman and Rao on the compression of bounded-round protocols with a direct sum result for information complexity. Since we can compute (albeit fairly inefficiently) the communication complexity of any function by enumerating all possible protocols of a certain length, this completes our proof.

The proof we provide below shows that zero-error internal information complexity is computable, but the same method (with a modification to Section \ref{sectic}; see Remark \ref{extremark}) also shows that zero-error external information complexity is computable. We believe similar techniques can be used to show that $\epsilon$-error information complexity is computable, but do not include such a proof in this paper. 

\subsection{Open Problems}

Naturally, the most immediate open problem arising from our work is understanding whether (and how much) the rate of convergence 
in Theorem~\ref{convthm} can be improved:

\begin{problem}\label{pr:1}
What is the (worst case) rate of convergence of the $r$-round information complexity of $f$ to $IC_\mu(f)$?
In other words, for a given $\epsilon>0$ and truth table size $N=|\mathcal{A}\times \mathcal{B}|$, how large
does $r(N,\epsilon)$ need to be to ensure that the $r$-round information complexity $IC_{r,\mu}(f)$ satisfies
$$
IC_{r,\mu}(f) > IC_{\mu}(f) - \epsilon?
$$
\end{problem}

In this paper we prove that $r(N,\epsilon)\le (N\epsilon^{-1})^{O(N)}$. On the other hand, \cite{BGPW} shows
that when $f$ is the two-bit $AND$ (and thus $N=4$ is a constant), the tight estimate for $r$ is $r=\Theta(\epsilon^{-1/2})$. 
Therefore, the polynomial dependence on $\epsilon$, even when $N$ is a constant, is necessary. On the other hand, 
we do not have any interesting lower bounds on $r$ in terms of $N$. In particular, it is not known whether 
the exponential dependence on $N$ is necessary here. 

The second open problem is in a similar vein, asking whether Theorem~\ref{mainthm} can be improved. 

\begin{problem}\label{pr:2}
What is the computational complexity of computing the (zero-error internal) information complexity of a function $f$
within error $\alpha$ given its truth table? By how much can the bound of  $2^{\exp\left((N\alpha^{-1})^{O(N)}\right)}$
be improved?
\end{problem}

By the analysis in Section~\ref{sectic}, any progress on Problem~\ref{pr:1} will translate into progress on Problem~\ref{pr:2}. 
For comparison, it is not hard to see that the trivial algorithm for computing the {\em average-case  communication complexity}
of a function $f:[n]\times [n] \rightarrow \{0,1\}$ (so that $N=n^2$) within an additive error $\alpha$ runs in time 
$2^{n\cdot N^{N/\alpha}} = 2^{\exp((N \alpha^{-1})^{O(1)})}$. In other words, there is an exponential gap between the trivial 
communication complexity upper bound and the bound we obtain in Theorem~\ref{mainthm}. 

\section{Preliminaries}\label{preliminaries}

\subsection{Information Theory}

We briefly review some standard information theoretic definitions used throughout this paper. For a more detailed introduction, we refer the reader to \cite{CT91}.

\begin{definition}[Entropy]
The \textit{entropy} of a random variable $X$ is $H(X) = \sum_{x} \mathrm{Pr}[X=x]\log(1/\mathrm{Pr}[X=x])$. The \textit{conditional entropy} $H(X|Y)$ is defined to be $\mathbb{E}_{y\sim Y}[H(X|Y=y)]$.
\end{definition}

\begin{definition}[Mutual Information]
The \textit{mutual information} between two random variables $A$, $B$, denoted $I(A;B)$ is defined to be the quantity $H(A) - H(A|B)$. The \textit{conditional mutual information} $I(A;B|C)$ is $H(A|C) - H(A|BC)$.
\end{definition}

\begin{definition}[Divergence]
The \textit{informational divergence} (also known as Kullback-Leibler distance or relative entropy) between two distributions $A$ and $B$ is 

$$D(A||B) = \sum_{x} A(x)\log(A(x)/B(x))$$
\end{definition}

\begin{proposition}[Chain Rule]
Let $C_1$, $C_2$, $D$, $B$ be random variables. Then $I(C_1C_2;B|D) = I(C_1;B|D) + I(C_2;B|C_1D)$. 
\end{proposition}

We will regularly make use of the following inequality for conditional mutual information. 

\begin{lemma}\label{subadd}
Let $A, B, C, D$ be four random variables such that $I(B;D|AC) = 0$. Then

\begin{equation*}
I(A;B|C) \geq I(A;B|CD)
\end{equation*}
\end{lemma}
\begin{proof}
We apply the chain rule twice:

\begin{eqnarray*}
I(A;B|CD) &=& I(AD;B|C) - I(D;B|C) \\
&=& I(A;B|C) + I(D;B|AC) - I(D;B|C) \\
&=& I(A;B|C) - I(D;B|C) \\
&\leq & I(A;B|C)
\end{eqnarray*}
\end{proof}

\subsection{Protocols and Information Complexity}

In the two-party communication setting, Alice is given an element $a$ from a finite set $\mathcal{A}$, while Bob is given an element $b$ from a finite set $\mathcal{B}$, where $(a,b)$ is drawn from some distribution $\mu$ over $\mathcal{A}\times \mathcal{B}$. Their goal is to compute $f(a,b)$, where $f: \mathcal{A}\times \mathcal{B} \rightarrow \{0,1\}$ is a function known to both parties. They would like to accomplish this while revealing as little information as possible; either to each other (in the case of information cost) or to an outside observer (in the case of external information cost). To do this, they execute a \textit{communication protocol}, which we view as being built out of \textit{signals}.

\begin{definition}\label{def:size}
A \emph{signal} $\sigma$ over a set $S$ is an assignment of a probability $\sigma_s \in [0,1]$ to each element $s$ in $S$. For a given element $s$ of $S$, we define $\sigma(s)$ to be the Bernoulli random variable that equals $1$ with probability $\sigma_{s}$. The \emph{size} of a signal $\sigma$ is given by $|\sigma| = \max_{s \in S}\left|\frac{1}{2} - \sigma_{s}\right|$. 
\end{definition}

\begin{definition}
A \emph{communication protocol} $\pi$ is a finite rooted binary tree, where each non-leaf node is labeled by either a signal over $\mathcal{A}$ (corresponding to Alice's move) or a signal over $\mathcal{B}$ (corresponding to Bob's move), and each edge is labeled either $0$ or $1$. Alice and Bob can execute this protocol by starting at the root and repeatedly performing the following procedure; if the signal $\sigma$ at the current node is a signal over $\mathcal{A}$, Alice sends Bob an instance of $\sigma(a)$, and they both move down the corresponding edge; likewise, if the signal is a signal over $\mathcal{B}$, Bob performs the analogous procedure. 

Each leaf node is labeled with a value $0$ or $1$. We say the communication protocol \emph{successfully computes} $f$ with zero error if the value of the leaf node Alice and Bob finish the protocol on is always equal to $f(a, b)$ for all $(a, b) \in \mathcal{A}\times\mathcal{B}$ (in particular, even $(a, b)$ where $\mu(a,b) = 0$). The communication cost $CC(\pi)$ of protocol $\pi$ is equal to the depth of the deepest leaf in $\pi$.
\end{definition}

This agrees with the usual definition of a private coins protocol (indeed, any bit Alice can ever send in any protocol must be a signal over $\mathcal{A}$, and likewise for Bob). A public coins protocol is simply a distribution over private coins protocols. For our purposes, it suffices to solely examine private coins protocols, since the information cost of a public coins protocol is simply the expected information cost of the corresponding private coins protocols. 

As is standard, we will let $A$ and $B$ be random variables representing Alice's input and Bob's input respectively, and let $\Pi$ be the random variable representing the protocol's transcript. We can then define the \textit{information cost} of a protocol and the \textit{information complexity} of a function as follows.

\begin{definition}
The \emph{information cost} of a protocol $\pi$ is given by

\begin{equation*}
IC_{\mu}(\pi) = I(A;\Pi|B) + I(B;\Pi|A)
\end{equation*}

\noindent
The \emph{external information cost} of a protocol $\pi$ is given by

\begin{equation*}
IC_{\mu}^{ext}(\pi) = I(AB;\Pi)
\end{equation*}
\end{definition}

\begin{definition}
The \emph{information complexity} of a function $f$ is given by

\begin{equation*}
IC_{\mu}(f) = \inf_{\pi} IC_{\mu}(\pi)
\end{equation*}

\noindent
where the infimum is over all protocols $\pi$ that successfully compute $f$. Likewise, the \emph{external information complexity} of a function $f$ is given by

\begin{equation*}
IC_{\mu}^{ext}(f) = \inf_{\pi} IC_{\mu}^{ext}(\pi)
\end{equation*}

\noindent
where again, the infimum is over all protocols $\pi$ that successfully compute $f$.
\end{definition}

Throughout the remainder of this paper, it will be useful to think of signals as operating on the space $\Delta(\mathcal{A}\times\mathcal{B})$ of probability distributions over $\mathcal{A}\times\mathcal{B}$, which we term \textit{beliefs}. At the beginning of a protocol, an outside observer's belief is simply given by $\mu$, the distribution $(a,b)$ was drawn from. As this observer observes new signals, his belief evolves according to Bayes' rule; for example, if he observes the signal $\sigma(a)$ sent by Alice, his belief changes from the prior belief $p$ to the posterior belief

\begin{equation}\label{shift0}
p_{0}(a,b) = \dfrac{(1-\sigma_{a})p(a,b)}{\sum_{i,j}(1-\sigma_{i})p(i,j)}
\end{equation}

\noindent
if $\sigma(a) = 0$ (which occurs with probability $P_0 = \sum_{i,j} (1-\sigma_{i})p(i,j)$) and to  the posterior belief 

\begin{equation}\label{shift1}
p_{1}(a,b) = \dfrac{\sigma_{a}p(a,b)}{\sum_{i,j}\sigma_{i}p(i,j)}
\end{equation}

\noindent
if $\sigma(a) = 1$ (which occurs with probability $P_1 = \sum_{i,j}\sigma_{i}p(i,j)$). As shorthand, we will say that $\sigma$ \textit{shifts} belief $p$ to $(p_0, p_1)$. Note that the probabilities $P_0$ and $P_1$ are uniquely recoverable given $p_0$ and $p_1$ (in particular, treating beliefs as vectors in $\mathbb{R}^{|\mathcal{A}\times\mathcal{B}|}$, it must be the case that $P_0p_0 + P_1p_1 = p$ and that $P_0 + P_1 = 1$).  If $P_0 = P_1 = \frac{1}{2}$, we say the signal is \textit{balanced} for the belief $p$ (when it is clear from context, we will omit which belief $p$ the signal is balanced for).

We can write similar equations that describe the change in beliefs upon observing the signal $\sigma(b)$ sent by Bob. An important consequence of equations \ref{shift0} and \ref{shift1} is that signals \emph{commute}. That is, sending signal $\sigma$ followed by signal $\sigma'$ results in the same probability distribution over beliefs as sending signal $\sigma'$ followed by signal $\sigma$. 

Given a protocol $\pi$, we can label all nodes of the protocol tree with the belief an observer would have at that point in the protocol. We can therefore alternatively express the information cost and external information cost of a protocol as a function of the final beliefs at the leaves of the protocol.

\begin{definition}
The \textit{information cost} at a node $v$ of protocol $\pi$ with belief $p=\mu_v$ is defined to be:

\begin{equation}\label{icost}
C(p) = \mathbb{E}_{a \sim p}[D(p(b|a)||\mu(b|a))] + \mathbb{E}_{b \sim p}[D(p(a|b)||\mu(a|b))]
\end{equation}

\noindent
The \textit{external information cost} at a node $v$ of protocol $\pi$ with belief $p$ is defined to be:

\begin{equation}\label{exticost}
C^{ext}(p) = D(p(a,b)||\mu(a,b))
\end{equation}
\end{definition}

\begin{lemma}\label{icleaves}
The information cost of a protocol is the expected value of the information cost at the leaves of the protocol. The external information cost of a protocol is the expected value of the external information cost at the leaves of the protocol.
\end{lemma}
\begin{proof}
We demonstrate the computation for information complexity; the computation for external information complexity is similar. Write $P(a,b,\pi)$ as shorthand for $\mathrm{Pr}[(A,B,\Pi) = (a,b,\pi)]$. Note that

\begin{eqnarray*}
I(A;\Pi|B) &=& \sum_{b} P(b)\sum_{a,\pi} P(a,\pi| b)\log\frac{P(a,\pi|b)}{P(a|b)P(\pi|b)} \\
&=& \sum_{b} P(b) \sum_{a,\pi} P(\pi|b)P(a|\pi, b)\log\frac{P(a|\pi, b)}{\mu(a|b)} \\
&=& \sum_{b, \pi}P(b)P(\pi|b) \sum_{a}P(a|\pi,b) \log \frac{P(a|\pi,b)}{\mu(a|b)} \\
&=& \sum_{b, \pi} P(\pi)P(b|\pi) D(P(a|\pi,b) || \mu(a|b)) \\
&=& \sum_{\pi} P(\pi) \sum_{b}P(b|\pi) D(P(a|\pi,b) || \mu(a|b)) \\
&=& \sum_{\pi} P(\pi) \mathbb{E}_{b\sim p} \left[ D(p(a|b) || \mu(a|b))\right] \\
\end{eqnarray*}

The last equality follows from the fact that, $P(a|b,\pi)$ is simply the belief about $a$ given $b$ at the leaf given by the transcript $\pi$, and hence is $p(a|b)$ (likewise, $P(b|\pi)$ equals $p(b)$ at that leaf). Combining this with the analogous equation for $I(B;\Pi|A)$, we find that $IC_{\mu}(\pi)$ is exactly the expected value of $C(p)$ over the leaves of the protocol, as desired.
\end{proof}

Alternatively, we can express the information cost of a protocol in terms of how much information each signal in the protocol leaks. 

\begin{definition}
Let $\sigma$ be a signal in protocol $\pi$ that shifts belief $p$ to $(p_0, p_1)$. Then, the \textit{information cost} $C(\sigma, p)$ of $\sigma$ is defined as

\begin{equation} \label{sigcostdef}
C(\sigma, p) = P_0C(p_0) + P_1C(p_1) - C(p)
\end{equation}

\noindent
The \textit{external information cost} $C^{ext}(\sigma, p)$ is similarly defined as

\begin{equation*}
C^{ext}(\sigma, p) = P_0C^{ext}(p_0) + P_1C^{ext}(p_1) - C^{ext}(p)
\end{equation*}
\end{definition}

\begin{lemma}\label{lemsigcost}
For each node $v$ in the protocol, let $p_{v}$ be the belief at node $v$, let $q_{v}$ be the probability of reaching node $v$, and let $\sigma_{v}$ be the signal we send at point $v$. Then the information cost of $\pi$ is equal to

\begin{equation*}
IC_{\mu}(\pi) = \sum_{v \in \pi} q_{v}C(\sigma_v, p_v)
\end{equation*}

\noindent
Likewise the external information cost of $\pi$ is equal to

\begin{equation*}
IC_{\mu}^{ext}(\pi) = \sum_{v \in \pi} q_{v}C^{ext}(\sigma_v, p_v)
\end{equation*}
\end{lemma}
\begin{proof}
Expanding each $C(\sigma_v, p_v)$ out according to equation \ref{sigcostdef}, all terms $C(p_v)$ for beliefs corresponding to non-terminal nodes $v$ in $\pi$ cancel out, and we are left with 

\begin{equation*}
\sum_{\mbox{leaf nodes } v} q_{v}C(p_{v})
\end{equation*}

\noindent
which is exactly the expected information cost at the leaves of $\pi$, which by Lemma \ref{icleaves} is equal to $IC_{\mu}(\pi)$, as desired. (A similar computation holds for the external information cost).
\end{proof}

\begin{remark}
Alternatively, one can show that (if $X_v$ is the output of signal $\sigma_{v}$ at node $v$) 

\begin{equation}
C(\sigma, p) = I(X_{v};A|B,\Pi_{pre}=v) + I(X_{v};B|A,\Pi_{pre}=v)
\end{equation}

\noindent
(one of these two terms will equal zero, depending on which party sends signal $\sigma_{v}$). Lemma \ref{lemsigcost} then follows from an application of the chain rule. 
\end{remark}

Throughout the remainder of the paper, we will let $N = |\mathcal{A}\times\mathcal{B}| = |\mathcal{A}|\cdot|\mathcal{B}|$. Note that $N$ is the size of the truth table of $f$ and is thus (in some sense) the size of the input to the problem of computing the information complexity of $f$. All logarithms are to base 2 unless otherwise specified. 

\section{Computability of Information Complexity}

\subsection{Restricting to $\mu$ with full support}\label{sectmu}

We begin by showing that we need only consider initial beliefs $\mu$ with full support; that is, where $\mu(x,y) >\rho> 0$ for all $x \in \mathcal{A}$, $y \in \mathcal{B}$. We accomplish this by showing we can perturb $\mu$ while only slightly changing the value of $IC_{\mu}(f)$. Recall that  $h:[0,1]\rightarrow [0,1]$ is Shannon's entropy function. 

\begin{theorem}\label{perturbmu}
Let $\mu \in \Delta(\mathcal{A}\times\mathcal{B})$ be a distribution without full support (i.e., $\mu(a,b) = 0$ for some $a \in \mathcal{A}$ and $b \in \mathcal{B}$). Let $\zeta \in \Delta(\mathcal{A}\times\mathcal{B})$ be the uniform distribution over pairs $(a,b)$ where $\mu(a,b) = 0$. Then, for any $\epsilon \in (0, 1)$, if $\tilde{\mu} = (1-\epsilon)\mu + \epsilon\zeta$, 

\begin{equation*}
\frac{1}{1-\epsilon}\left(IC_{\tilde{\mu}}(\pi) - 2h(\epsilon) - \epsilon\log N\right) \leq IC_{\mu}(f) \leq \frac{1}{1-\epsilon}IC_{\tilde{\mu}}(f)
\end{equation*}
\end{theorem}
\begin{proof}
Fix a protocol $\pi$ that successfully computes $f$. Let $Z$ be a Bernoulli random variable with probability $\epsilon$. Note that we can sample from $\tilde{\mu}$ by sampling from $\mu$ if $Z = 0$ and sampling from $\zeta$ if $Z=1$. Letting $I_{\mu}(\Pi;A|B)$ denote $I(\Pi;A|B)$ when $(A,B)$ is distributed according to $\mu$, we have that

\begin{eqnarray*}
I_{\mu}(\Pi;A|B) &=& I_{\tilde{\mu}}(\Pi;A|BZ=0) \\
&=& \frac{1}{1-\epsilon}\left(I_{\tilde{\mu}}(\Pi;A|BZ) - \epsilon I_{\tilde{\mu}}(\Pi;A|BZ=1)\right) \\
&\leq & \frac{1}{1-\epsilon}I_{\tilde{\mu}}(\Pi;A|BZ) \\
&\leq & \frac{1}{1-\epsilon}I_{\tilde{\mu}}(\Pi;A|B)
\end{eqnarray*}

\noindent
where the last inequality follows from Lemma \ref{subadd} since $I(\Pi;Z|AB) = 0$. Combining this with the corresponding calculation for $I_{\mu}(\Pi;B|A)$, we see that 

\begin{equation}\label{icmulb}
IC_{\mu}(\pi) \leq \frac{1}{1-\epsilon}IC_{\tilde{\mu}}(\pi)
\end{equation}

On the other hand, note that $I_{\tilde{\mu}}(\Pi;A|BZ=1) \leq H(A) \leq \log |\mathcal{A}|$. From this, we see that

\begin{eqnarray*}
I_{\mu}(\Pi;A|B) &=& \frac{1}{1-\epsilon}\left(I_{\tilde{\mu}}(\Pi;A|BZ) - \epsilon I_{\tilde{\mu}}(\Pi;A|BZ=1)\right) \\
&\geq & \frac{1}{1-\epsilon}\left(I_{\tilde{\mu}}(\Pi;A|BZ) -\epsilon \log |\mathcal{A}| \right)\\
&\geq & \frac{1}{1-\epsilon}\left(I_{\tilde{\mu}}(\Pi;A|B) - H(Z) - \epsilon \log |\mathcal{A}| \right)\\
&\geq & \frac{1}{1-\epsilon}\left(I_{\tilde{\mu}}(\Pi;A|B) - h(\epsilon) - \epsilon \log |\mathcal{A}| \right)
\end{eqnarray*}

\noindent
Combining this with the corresponding calculation for $I_{\mu}(\Pi;B|A)$, we see that

\begin{equation}\label{icmuub}
IC_{\mu}(\pi) \geq \frac{1}{1-\epsilon}IC_{\tilde{\mu}}(\pi) - \frac{2h(\epsilon)}{1-\epsilon} - \frac{\epsilon}{1-\epsilon}\log N
\end{equation}

By taking the infimum of both sides of equations \ref{icmulb} and \ref{icmuub} over all protocols $\pi$ that successfully compute $f$, we obtain the desired result.

\end{proof}

As a corollary of Theorem \ref{perturbmu}, to compute $IC_{\mu}(f)$ to within $\alpha$, it suffices to choose $\epsilon$ in the above theorem so that $\frac{1}{1-\epsilon}(2h(\epsilon) + \epsilon\log N) < \frac{\alpha}{2}$ (so that $(1-\epsilon)^{-1}IC_{\tilde{\mu}}(f)$ is within $\frac{\alpha}{2}$ of $IC_{\mu}(f)$), and then compute $IC_{\tilde{\mu}}(f)$ to within an additive error of $(1-\epsilon)\frac{\alpha}{2}$.

For the remainder of the proof, we will therefore assume that $\mu$ has full support, and define $\rho = \min_{a,b}\mu(a,b)$. Note that, by the proof of Theorem \ref{perturbmu}, we can always ensure that $\rho = \Omega \left(\frac{\alpha^2}{N\log N}\right)$. 

\subsection{Using signals far from the boundary} \label{sectboundary}

We next show that we can restrict our attention to protocols where the belief at each node is sufficiently separated from the boundary of $\Delta(\mathcal{A}\times\mathcal{B})$. 

\begin{definition}
A signal $\sigma$ is a \textit{revealer} if there exists an $i$ such that $\sigma_{i} = 1$ and $\sigma_{j} = 0$ for all $j \neq i$.
\end{definition}

\begin{definition}
A belief $p$ is \emph{$\gamma$-safe} if, for all $a\in \mathcal{A}$ and $b \in \mathcal{B}$, either $p(a,b) \geq \gamma$ or $p(a,b) = 0$. A protocol $\pi$ is \emph{$\gamma$-safe} if the signal at every node in $\pi$ without a $\gamma$-safe belief is a revealer.
\end{definition}

\begin{theorem}\label{thmboundary}
Let $\pi$ be a communication protocol with information cost $C$. Then, for all $\gamma \in (0, 1)$, there exists a $\gamma$-safe protocol $\pi'$ that computes the same function as $\pi$ that has information cost at most $C + (|\mathcal{A}|+|\mathcal{B}|)h(\rho^{-1}\sqrt{\gamma})$. 
\end{theorem}

We will make use of the following two lemmas.

\begin{lemma}\label{rtgamma}
If at some point in a protocol $\pi$, $ p(a,b) < \gamma$, then either the probability that $A=a$ or the probability that $B=b$ must be small:

\begin{equation*}
\min(p_{A}(a), p_{B}(b)) < \rho^{-1}\sqrt{\gamma}
\end{equation*}

\end{lemma}
\begin{proof}
View the belief $p$ as a $|\mathcal{A}|$ by $|\mathcal{B}|$ matrix of real numbers. Note that each time Alice sends a signal, she updates this belief by multiplying each row of this matrix by a different number; likewise, every time Bob sends a signal, he updates this belief by multiplying each column of this matrix by a different number.

At this point in the protocol, for each $a \in \mathcal{A}$, let $\lambda_{a}$ be the product of all the updates to row $a$; likewise, let $\kappa_{b}$ be the product of all the updates to column $b$. It follows that 

\begin{equation*}
p(a,b) = \mu(a,b)\lambda_{a}\kappa_{b}
\end{equation*}

\noindent
Likewise, we can write

\begin{eqnarray}
p_{A}(a) &=& \lambda_{a}\sum_{j}\mu(a,j)\kappa_{j} \label{proba}\\
p_{B}(b) &=& \kappa_{b}\sum_{i}\mu(i,b)\lambda_{i} \label{probb}
\end{eqnarray}

\noindent
Multiplying equations \ref{proba} and \ref{probb}, we obtain

\begin{equation}  \label{prodab}
p_{A}(a)p_{B}(b) = \lambda_{a}\kappa_{b}\sum_{i,j}\mu(i,b)\mu(a,j)\lambda_{i}\kappa_{j}
\end{equation}

\noindent
Finally, note that (since $\sum_{i,j} p(i,j) = 1$), we have that

\begin{equation} \label{probsum}
\sum_{i,j}\lambda_{i}\kappa_{j}\mu(i,j) = 1
\end{equation}

\noindent
It follows from equations \ref{prodab}, and \ref{probsum} that

\begin{eqnarray*}
p_{A}(a)p_{B}(b)  &=&  \lambda_{a}\kappa_{b}\sum_{i,j}\mu(i,b)\mu(a,j)\lambda_{i}\kappa_{j} \\
&=& \lambda_{a}\kappa_{b}\sum_{i,j}\mu(a,b)\mu(i,j)\frac{\mu(i,b)\mu(a,j)}{\mu(a,b)\mu(i,j)}\lambda_{i}\kappa_{j}\\
&\leq & \rho^{-2}\mu(a,b)\lambda_{a}\kappa_{b}\sum_{i,j}\mu(i,j)\lambda_{i}\kappa_{j} \\
&=& \rho^{-2}\mu(a,b)\lambda_{a}\kappa_{b} \\
&=& \rho^{-2}p(a,b) \\
&< & \rho^{-2}\gamma
\end{eqnarray*}

\noindent
Therefore, $\min(p_{A}(a), p_{B}(b)) < \sqrt{\rho^{-2}\gamma} = \rho^{-1}\sqrt{\gamma}$, as desired.
\end{proof}

\begin{lemma}\label{reveal}
Let $\pi$ be a protocol with information cost $C$. Let $v$ be a node in this protocol with belief $p$. If, at $v$, Alice reveals whether $a=i$ (with the rest of the protocol remaining unchanged), then this modified protocol has information cost at most $C + h(p_{A}(i))$.  (Here $h$ is the binary entropy function). 

The analogous statement holds for Bob.
\end{lemma}
\begin{proof}
Alice can reveal whether $a=i$ by sending the revealer signal $\sigma$ where $\sigma_{i} = 1$ and $\sigma_{j} = 0$ for $j\neq i$. Since signals commute, Alice can equivalently reveal whether $a=i$ at the end of the protocol (assuming she passed through node $v$) instead of right after $v$.

At the end of the protocol (but before Alice reveals whether $a=i$), there may be multiple possible terminal beliefs; label these beliefs $p_1$ through $p_K$, with belief $p_i$ occurring with probability $Q_i$. Since these are the terminal beliefs that are descendants of node $v$, it follows that

\begin{equation*}
\sum_{k=1}^{K} Q_{k}p_{k} = p
\end{equation*}

In particular, $\sum Q_{k}p_{k,A}(i) = p_{A}(i)$. Now, as a consequence of equation \ref{icost}, revealing whether $a=i$ while at belief $p_{k}$ increases the expected information cost of the node by 

\begin{eqnarray}
\sum_{b}p_{k,B}(b)h(p_{k,A|B}(i|b)) &=& \sum_{b} p_{k,B}(b)h\left(\dfrac{p_{k}(i,b)}{p_{k,B}(b)}\right) \\
&\leq & h\left(\sum_{b}p_{k}(i,b)\right)\\
&=& h(p_{k,A}(i))
\end{eqnarray}

\noindent
where the inequality follows from Jensen's inequality, since $h(x)$ is concave. It follows that the total expected increase in the information cost of this protocol is at most 

\begin{equation*}
\sum_{k=1}^{K} Q_{k}h(p_{k,A}(i)) \leq h\left(\sum_{k=1}^{K}Q_{k}p_{k,A}(i)\right) = h(p_{A}(i))
\end{equation*}

\noindent
where the first inequality again follows from Jensen's inequality. This completes the proof. \end{proof}

We can now complete the proof of Theorem \ref{thmboundary}.

\begin{proof}[Proof of Theorem \ref{thmboundary}]
We will construct $\pi'$ from $\pi$ in the following manner: follow $\pi$ until you reach a belief $p$ satisfying $p(i,j) < \gamma$ for some choice of $i$ and $j$. Then, by Lemma \ref{rtgamma}, either $p_{A}(i)$ or $p_{B}(j)$ is at most $\rho^{-1}\sqrt{\gamma}$. Without loss of generality, assume $p_{A}(i) < \rho^{-1}\sqrt{\gamma}$. Then, Alice will reveal whether $a = i$. We repeat this process until the resulting protocol is $\gamma$-safe.

Note that on any complete path through $\pi'$, Alice and Bob perform at most $|\mathcal{A}| + |\mathcal{B}|$ reveals (since each reveal eliminates at least one of the $|\mathcal{A}|$ options for $a$ or the $|\mathcal{B}|$ options for $b$). By Lemma \ref{reveal}, this means the information cost of $\pi'$ is at most $(|\mathcal{A}|+|\mathcal{B}|)h(\rho^{-1}\sqrt{\gamma})$ larger than the information cost of $\pi$, as desired.
\end{proof}


\subsection{Using signals of bounded size}\label{sectsize}
We next show that we can restrict our attention to protocols that only use signals of a bounded size. Here, by a bounded size, we require both that each individual component of the signal is sufficiently small and that the amount the signal shifts the corresponding belief is sufficiently large. 

\begin{definition}
A signal $\sigma$ that shifts $p$ to $(p_0, p_1)$ has \textit{power} $d$ at belief $p$ if

\begin{equation*}
d = \max\left(||p-p_{0}||_\infty,\, ||p-p_{1}||_\infty\right)
\end{equation*}

\noindent
(When it is clear from context, we will often omit the specific belief $p$).
\end{definition}

Recall that a signal is balanced if the probability $P_0$ it is $0$ is equal to $P_1=1/2$. 
Also recall that by Definition~\ref{def:size}, the size of a signal 
is its maximum input-wise deviation from $1/2$ given by $|\sigma| = \max_{s \in S}\left|\frac{1}{2} - \sigma_{s}\right|$.
We prove:

\begin{theorem}\label{thmbounded}
Let $\pi$ be a $\gamma$-safe communication protocol with information cost $C$. Then, for every $\epsilon > 0$, there exists a $\gamma$-safe communication protocol $\pi'$ that computes the same function as $\pi$ with information cost at most $C+\epsilon$, but that only uses (in addition to revealer signals) balanced signals of size at most $\gamma^{-1}\delta$ and power at least $\delta$, for some positive $\delta$.
\end{theorem}

\begin{remark}
Note that unlike $\gamma$ which is (an explicit, easily computable) function of $\epsilon$, we do not assert anything about the computability 
of $\delta$ in Theorem~\ref{thmbounded}. We only need to know that such a $\delta$ exists. The dependence on $\delta$ will be removed 
later in the analysis. 
\end{remark}

To prove Theorem \ref{thmbounded}, we will make use of two lemmas. The first lemma provides a connection between the size and power of a ($\gamma$-safe) signal.

\begin{lemma} \label{sizepower}
If signal $\sigma$ is balanced and has power at most $2\delta$ at a $\gamma$-safe belief $p$, then $|\sigma| \leq \frac{\delta}{\gamma}$.
\end{lemma}
\begin{proof}
Since $\sigma$ is balanced at $p$, $P_0 = P_1 = \frac{1}{2}$, so it follows from equations \ref{shift0} and \ref{shift1} that

\begin{eqnarray} 
p_{0}(a,b) &=& 2p(a,b)(1-\sigma_{a,b}) \label{balshift0}\\ 
p_{1}(a,b) &=& 2p(a,b)\sigma_{a,b} \label{balshift1}
\end{eqnarray}

\noindent
(if Alice is sending this signal, then $\sigma_{a,b} = \sigma_{a}$; similarly, if Bob is sending this signal, then $\sigma_{a,b} = \sigma_{b}$). From equation \ref{balshift1}, we see that

\begin{equation*}
\sigma_{a,b} = \dfrac{p_{1}(a,b)}{2p(a,b)}
\end{equation*}

\noindent
and in particular,

\begin{equation*}
\left|\sigma_{a,b} - \frac{1}{2}\right| = \dfrac{|p_{1}(a,b) - p(a,b)|}{2p(a,b)}
\end{equation*}

Since belief $p$ is $\gamma$-safe, $p(a,b) \geq \gamma$, and since $\sigma$ has power at most $2\delta$ at $p$, $|p_{1}(a,b) - p(a,b)| \leq 2\delta$. It follows that


\begin{equation*}
\left|\sigma_{a,b} - \frac{1}{2}\right| \leq \frac{\delta}{\gamma}
\end{equation*}

\noindent
and therefore that $|\sigma| \leq \frac{\delta}{\gamma}$, as desired.
\end{proof}

The second lemma allows us to `decompose' a signal into a sequence of smaller subsignals.

\begin{definition}
Let $\sigma$ be a signal that shifts the belief $p$ to $(p_0, p_1)$. A signal $\sigma'$ that shifts $q$ to $(q_0, q_1)$ is a \textit{subsignal} of $\sigma$ if $q$ lies on the segment connecting $p_0$ and $p_1$, $q_{0}$ lies on the segment connecting $q$ and $p_0$, and $q_1$ lies on the segment connecting $q$ and $p_1$.
\end{definition}

\begin{lemma}
If Alice can send a signal $\sigma$, then she can also send all subsignals of $\sigma$. The analogous statement holds for Bob.
\end{lemma}
\begin{proof}
First, note that since $q$ is a convex combination of $q_0$ and $q_1$, there is some signal that shifts $q$ to $(q_0, q_1)$. Recall that Alice can send any signal that satisfies $\sigma_{a,b} = \sigma_{a,b'}$ for all $b, b' \in \mathcal{B}$ and all $a \in \mathcal{A}$. Since $\sigma_{a,b} = \frac{p_{1}(a,b)}{p(a,b)}\sum_{i,j}\sigma_{i,j}p(i,j)$ (equation \ref{shift0}), we have that

\begin{equation*}
\frac{p_{1}(a,b)}{p(a,b)} = \frac{p_{1}(a,b')}{p(a,b')}
\end{equation*}

\noindent
Since $q_1$ and $q$ are linear combinations of $p_1$ and $p$, it follows that

\begin{equation*}
\frac{q_{1}(a,b)}{q(a,b)} = \frac{q_{1}(a,b')}{q(a,b')}
\end{equation*}

\noindent
and therefore that $\sigma'_{a,b} = \sigma'_{a,b'}$. It follows that Alice can send signal $\sigma'$. 
\end{proof}

We can now proceed to prove Theorem \ref{thmbounded}.

\begin{proof}[Proof of Theorem \ref{thmbounded}]
Let $p_{min}$ be the minimum power of a signal in $\pi$. We will choose $\delta$ to equal $\min\left(\frac{p_{min}}{10}, \frac{\gamma}{10}\right)$. 

Let $\sigma$ be an arbitrary signal in $\pi$ that shifts the belief $p$ to $(p_{0}, p_{1})$. We will replace $\sigma$ with the following `subprotocol'. Intuitively, the following subprotocol uses several small signals of power roughly $\delta$ to perform a random walk on the segment between beliefs $p_{0}$ and $p_{1}$, terminating when it hits one of the two boundary beliefs. Since this protocol ensures that the belief $p$ evolves to either belief $p_0$ or belief $p_1$ (with the corresponding uniquely determined probabilities), it accomplishes the same effect on the distribution of beliefs as sending signal $\sigma$.

More specifically, we can describe the subprotocol as follows. Assume our current belief $q$ lies on the segment between $p_{0}$ and $p_{1}$. Compute $d$, the $L_{\infty}$ distance from $q$ to the nearest endpoint (that is, $d = \min\left(||q - p_{0}||_\infty, ||q - p_{1}||_\infty\right)$). If $d \leq 2\delta$, send a balanced subsignal of $\sigma$ of power $d$; this either sends $q$ to the nearest endpoint, or increases $d$ to $2d$ (since $||p_1 - p_0||_{\infty} \geq p_{min} \geq 10\delta$). On the other hand, if $d > 2\delta$, simply send a balanced subsignal of $\sigma$ of power $\delta$; this decreases $d$ by at most $\delta$.

It is straightforward to see that in the above subprotocol, we only ever send balanced signals of power between $\delta$ and $2\delta$ (in particular, we never get closer than $\delta$ to an endpoint until we reach it). Since $\pi$ is $\gamma$-safe, it follows from Lemma \ref{sizepower} that each signal we use in this subprotocol also has size at most $\gamma^{-1}\delta$. 


Since the $L_\infty$ distance between $p_0$ and $p_1$ is at most $1$, this random walk will terminate with probability $1$ in finite time. Unfortunately, our resulting protocol is no longer finite. We can remedy this by adjusting our subprotocol so that after some large number $T$ of steps, the two parties abort the protocol and simply exchange both of their inputs. This ensures the two parties can successfully compute the function $f$ but potentially increases the information cost of the protocol. Since the information cost at any node of the protocol is bounded above (by $\log|\mathcal{A}| + \log|\mathcal{B}|$) and since the probability we have to abort our subprotocol decreases in $T$, by choosing a sufficiently large value of $T$ we can ensure, for any $\epsilon > 0$, that this modified protocol has information cost at most $C + \epsilon$. 
\end{proof}

\subsection{Using a bounded number of signals}\label{sectnum}
We now show that we can convert any protocol into a protocol that only uses a bounded number of distinct signals, while only increasing the information leaked by a small additive factor. 

\begin{theorem}\label{thmsignals}
Let $\pi$ be a $\gamma$-safe communication protocol with information cost $C$ that only uses (in addition to revealer signals) balanced signals of size at most $\gamma^{-1}\delta$ and power at least $\delta$. Then, for any $\epsilon > 0$, there exists a communication protocol $\pi'$ that computes the same function as $\pi$ with information cost at most $C+\epsilon$ but that only uses $Q$ different signals, where

\begin{equation*}
Q =  \left(\frac{648}{\epsilon\gamma^3\ln 2}\right)^{N/2} + (|\mathcal{A}| + |\mathcal{B}|)
\end{equation*}
\end{theorem}

To show this, we first argue that signals that are close component-wise have similar effects when they act on beliefs. 

\begin{lemma}\label{sigerror}
Let $\sigma$ and $\sigma'$ be two signals over a set of size $N$ such that $\sigma$ is balanced at belief $p$, and for each $i$, $|\sigma_i - \sigma'_i| < \epsilon$ (for some $\epsilon < 1/6$). Let $\sigma$ shift belief $p$ to $(p_0, p_1)$ and $\sigma'$ shift belief $p$ to $(p'_0, p'_1)$. Then, as elements of $\mathbb{R}^N$, $||p_0 - p'_0||_2 \leq 9\epsilon$ and $||p_1 - p'_1||_2 \leq 9\epsilon$. 
\end{lemma}
\begin{proof}
Recall that

\begin{eqnarray*}
p_1(i) &=& \dfrac{\sigma_ip(i)}{\sum_{i} \sigma_i p(i)}\\
p'_1(i) &=& \dfrac{\sigma'_ip(i)}{\sum_{i} \sigma'_i p(i)}\\
\end{eqnarray*}

\noindent
Moreover, note that since $\sigma$ is balanced at $p$, $\sum_{i} \sigma_i p(i) = \frac{1}{2}$. Now, since $|\sigma_i - \sigma'_i| < \epsilon$ for all $i$, we have that

\begin{equation*}
\frac{1}{2} - \epsilon \leq \sum_{i} \sigma'_i p(i) \leq \frac{1}{2}+\epsilon
\end{equation*}

\noindent
It follows that

\begin{eqnarray*}
p_1(i) - p'_1(i) &\leq & \frac{\sigma_{i}p(i)}{1/2} - \frac{(\sigma_{i}-\epsilon)p(i)}{1/2 + \epsilon} \\
&=& \frac{\sigma_{i}p(i)(1/2 + \epsilon) - (1/2)(\sigma_{i}-\epsilon)p(i)}{(1/2)(1/2+\epsilon)} \\
&=& \frac{\sigma_{i}p(i)\epsilon + (1/2)p(i)\epsilon}{(1/2)(1/2+\epsilon)}\\
&=& \left(\frac{\sigma_{i}}{(1/2)(1/2+\epsilon)} + \frac{1}{1/2+\epsilon}\right)\epsilon p(i) \\
&\leq & (4+2)\epsilon p(i)\\
&=& 6\epsilon p(i)
\end{eqnarray*}

\noindent
(where this last inequality follows from the fact that $\sigma_i$ is less than $1$). Likewise,

\begin{eqnarray*}
p'_1(i) - p_1(i) &\leq & \frac{(\sigma_{i}+\epsilon)p(i)}{1/2 - \epsilon} - \frac{\sigma_{i}p(i)}{1/2} \\
&=& \frac{(1/2)(\sigma_{i}+\epsilon)p(i) - (1/2 - \epsilon)\sigma_{i}p(i)}{(1/2)(1/2-\epsilon)} \\
&=& \frac{\sigma_{i}p(i)\epsilon + (1/2) p(i)\epsilon}{(1/2)(1/2 - \epsilon)} \\
&=& \left(\frac{\sigma_{i}}{(1/2)(1/2-\epsilon)} + \frac{1}{1/2-\epsilon}\right)\epsilon p(i) \\
&\leq & (6+3)\epsilon p(i) \\
&= & 9\epsilon p(i)
\end{eqnarray*}

It follows that $|p'_1(i) - p_1(i)| \leq 9\epsilon p(i)$, and therefore that $||p_1 - p'_1||_{2} \leq 9\epsilon\sqrt{\sum |p(i)|^2} \leq 9\epsilon$. The proof for $p_0$ and $p'_0$ follows similarly.
\end{proof}

Our main strategy for reducing the number of distinct signals used by our protocol is to choose a dense set $S$ of signals and replace each signal in our protocol with a nearby close signal in $S$. The following lemma bounds the additional information leaked by this replacement procedure.

\begin{lemma}\label{replacesig}
Let $\sigma$ be a signal in protocol $\pi$ that shifts belief $p$ to $(p_0, p_1)$. Write $\sigma$ as the convex combination $\sum_{i=1}^{n}w_{i}\sigma^{(i)}$ of $n$ signals $\sigma^{(i)}$ (with all $w_{i} \in [0,1]$ and $\sum w_{i} = 1$). Let $q_{\sigma}$ be the probability that signal $\sigma$ is sent as part of protocol $\pi$ (i.e., the probability we reach the corresponding node of $\pi$), and let $\pi_i$ be the protocol obtained by replacing signal $\sigma$ with signal $\sigma^{(i)}$. Then for some $i$,

\begin{equation*}
IC_{\mu}(\pi_i) \leq IC_{\mu}(\pi) + q_{\sigma}\left(\left(\sum_{i=1}^{n}w_{i}C\left(\sigma^{(i)}, p\right)\right) - C(\sigma, p)\right)
\end{equation*}

\end{lemma}
\begin{proof}
Without loss of generality, let us assume that Alice is sending signal $\sigma$. Let us consider two possibilities for Alice's action when she is about to send signal $\sigma$.

In the first case, she chooses a signal $\sigma^{(i)}$ randomly with probability $w_{i}$ and sends that signal. This is equivalent to just sending signal $\sigma$ (in particular, the probability we send a $1$ is $\sum_{i}w_{i}\sigma^{(i)}_{a} = \sigma_{a}$), and altogether, this is equivalent to executing the original protocol.

In the second case, she chooses a signal $\sigma^{(i)}$ randomly with probability $w_{i}$ and sends that signal, along with the index $i$ that she chose. This is equivalent to choosing a protocol $\pi_{i}$ randomly with probability $p_i$ and executing that protocol (in particular, she can choose the index $i$ at the beginning of the protocol). Note that, since $\pi$ is a zero-error protocol that computes $f$, $\pi_i$ must also compute $f$ with zero-error, so $\pi_i$ is also a valid protocol for this problem.

Let $K$ be the random variable corresponding to the index that Alice chooses (if signal $\sigma$ is never sent, then $K=-1$), and as before, let $\Pi$ be the random variable corresponding to the transcript of $\pi$. The total information Alice reveals to Bob in the first case is then $I(\Pi;A|B)$, and the (expected) total information Alice reveals to Bob in the second case is then $I(\Pi K;A|B)$. 

Divide $\Pi$ into two parts; $\Pi_{pre}$, which contains the transcript of $\pi$ up to and including the transmission of $\sigma$, and $\Pi_{fin}$, which contains the remainder of the transcript after the index $K$ is revealed. Note that, by Lemma \ref{lemsigcost}, 

\begin{equation*}
I(\Pi_{pre} K;A|B) - I(\Pi_{pre};A|B) = q_{\sigma}\left(\left(\sum_{i=1}^{n}w_{i}C\left(\sigma^{(i)},p\right)\right) - C(\sigma,p)\right)
\end{equation*}

Since $K$ and $\Pi_{fin}$ are conditionally independent given $A$, $B$ and $\Pi_{pre}$, i.e. $I(K; \Pi_{fin}|\Pi_{pre} A B)=0$, by Lemma \ref{subadd}, it follows that

\begin{eqnarray*}
I(\Pi K; A|B) - I(\Pi; A|B) &=& I(K;A|\Pi B) \\
&=& I(K;A|\Pi_{pre}B \Pi_{fin}) \\
&\leq & I(K;A|\Pi_{pre}B) \\
&= & I(\Pi_{pre} K;A|B) - I(\Pi_{pre};A|B) \\
& = & q_{\sigma}\left(\left(\sum_{i=1}^{n}w_{i}C\left(\sigma^{(i)},p\right)\right) - C(\sigma,p)\right)
\end{eqnarray*}

\noindent
Since $\mathbb{E}\left[IC_{\mu}(\pi_i)\right] - IC_{\mu}(\pi) = I(\Pi K; A|B) - I(\Pi; A|B)$, the result follows. 

\end{proof}

Finally, we use the continuity properties of the cost function $C$ to effectively bound the quantities in Lemma \ref{replacesig}.

\begin{lemma}\label{costub}
Let $f: \mathbb{R}^{N} \rightarrow \mathbb{R}$ be a function that is smooth on a convex compact subset $R$ of $\mathbb{R}^{N}$. Let $x$ be a point in $R$. Let $x_1, \dots, x_k \in R$ and $w_1, \dots, w_k \in [0, 1]$ satisfy $\sum_{i} w_{i}x_{i} = x$, $\sum_{i} w_{i} = 1$, and $||x-x_{i}|| \leq \epsilon$ for all $i$ (here $\|\cdot\|$ is the standard euclidean norm). Then

\begin{equation*}
\left| f(x) - \sum_{i=1}^{k} w_{i}f(x_i)\right| \leq U\epsilon^2
\end{equation*}

\noindent
where 

\begin{equation*}
U = \max_{z \in R} \left|\lambda_{max}(D^{2}f(z))\right|
\end{equation*}

\noindent
where $\lambda_{max}(M)$ is the largest eigenvalue (by absolute value) of $M$, and $D^2f(a)$ is the Hessian of $f$ at $a$.
\end{lemma}
\begin{proof}
For each $i$, let $v_{i} = x_{i} - x$. By the Taylor expansion of $f$ (with the mean-value form of the remainder), we know that for any $x$ in $R$,

\begin{equation*}
f(x+v) = f(x) + v^{t}Df(x) + v^{t}D^2f(y)v
\end{equation*}

\noindent
for some $y$ on the line segment connecting $x$ and $x+v$. Since $\sum w_{i} = 1$ and $\sum w_{i}v_{i} = 0$, it follows that

\begin{equation*}
\sum_{i=1}^{k} w_{i}f(x_i) = f(x) + \sum_{i=1}^{k}w_{i}v_{i}^{t}D^2f(y_i)v_{i}
\end{equation*}

\noindent
for some $y_i$ on the line segment connecting $x$ and $x_i$. Since $|v^{t}Mv| \leq |\lambda_{max}(M)|\cdot||v||^2$, $|v_{i}^{t}D^2f(y)v_{i}| \leq U\epsilon^2$ for all $i$. It follows that 

\begin{equation*}
\left|\sum_{i=1}^{k}w_{i}v_{i}^{t}D^2f(y_i)v_{i}\right| \leq \sum_{i=1}^{k}w_{i}|v_{i}^{t}D^2f(y_i)v_{i}| \leq U\epsilon^2
\end{equation*}

\noindent
and therefore that

\begin{equation*}
\left| f(x) - \sum_{i=1}^{k} w_{i}f(x_i)\right| \leq U\epsilon^2
\end{equation*}

\end{proof}

It is straightforward to verify that the cost function $C(p)$ is smooth over the region $R_{\gamma}$ given by $p(i,j) \geq \gamma$  and thus satisfies the condition of Lemma \ref{costub}. Moreover, we can compute explicit upper bounds for the constant $U_{\gamma}$ for this function. We compute one such bound below.

\begin{lemma}\label{ugammabound}
Let $R_{\gamma}$ be the subset of $\mathbb{R}^{|\mathcal{A}\times\mathcal{B}|} = \mathbb{R}^{N}$ defined by $p(a,b) \in [\gamma, 1]$ for all $a \in \mathcal{A}$ and $b \in \mathcal{B}$ (in particular, we \textbf{do not} have the constraint that $\sum_{a,b} p(a,b) = 1$). Then if
\begin{equation*}
U_{\gamma} = \max_{z \in R_{\gamma}}\left|\lambda_{max}(D^2C(p))\right|
\end{equation*}

\noindent
we have that $U_{\gamma} \leq (2/\ln 2)\gamma^{-1}$.  
\end{lemma}
\begin{proof}
Recall (equation \ref{icost}) that $C(p) = \mathbb{E}_{a \sim p}[D(p(b|a)||\mu(b|a))] + \mathbb{E}_{b \sim p}[D(p(a|b)||\mu(a|b))]$. Write

\begin{eqnarray*}
C_{A}(p) &=& \mathbb{E}_{a \sim p}[D(p(b|a)||\mu(b|a))] \\
C_{B}(p) &=& \mathbb{E}_{b \sim p}[D(p(a|b)||\mu(a|b))]
\end{eqnarray*}

\noindent
Note that we can write

\begin{eqnarray*}
C_{A}(p) &=&  \mathbb{E}_{a \sim p}[D(p(b|a)||\mu(b|a))] \\
&=& \sum_{a} p(a)D(p(b|a)||\mu(b|a)) \\
&=& \sum_{a,b} p(a)p(b|a)\log\frac{p(b|a)}{\mu(b|a)} \\
&=& \sum_{a,b} p(a,b)\left(\log p(a,b) - \log p(a) - \log \mu(a,b) + \log \mu(a)\right) \\
&=& \left(\sum_{a,b} p(a,b)\log p(a,b)\right) - \left(\sum_{a}p(a)\log p(a)\right) \\
&-& \left(\sum_{a,b}p(a,b)\log\mu(a,b)\right) + \left(\sum_{a}p(a)\log\mu(a)\right) 
\end{eqnarray*}

\noindent
Let $D_{a,b}$ stand for $\frac{\partial}{\partial p(a,b)}$. Then, it follows that (over $p \in R_{\gamma}$):

\begin{eqnarray*}
D_{a,b}C_{A} &=& \log p(a,b) - \log p(a) - \log \mu(a,b) + \log \mu(a)\\
D_{a,b}^2C_{A} &=& \frac{1}{\ln 2}\left(\frac{1}{p(a,b)} - \frac{1}{p(a)}\right) \\
|D_{a,b}^2C_{A}| &\leq & (\ln 2)^{-1}\gamma^{-1} \\
D_{a,b'}D_{a,b}C_{A} &=& -\frac{1}{p(a)\ln 2} \\
|D_{a,b'}D_{a,b}C_{A}| &\leq &  (\ln 2)^{-1}\gamma^{-1} \\
D_{a',b}D_{a,b}C_{A} &=& 0 \\
D_{a',b'}D_{a,b}C_{A} &=& 0 
\end{eqnarray*}

\noindent
with similar equations for $C_{B}$. It follows that the maximum entry (by absolute value) of $D^2 C(p)$ for $p$ in the region $R_{\gamma}$ is bounded above by $(2/\ln 2)\gamma^{-1}$. Since the largest eigenvalue of a matrix is bounded above by the largest entry in the matrix, this implies the desired bound on $U_{\gamma}$.

\end{proof}

We can now proceed to prove Theorem \ref{thmsignals}.

\begin{proof}[Proof of Theorem \ref{thmsignals}]
Let $U_{\gamma} = (2/\ln 2)\gamma^{-1}$, and set $M = \sqrt{\frac{81U_{\gamma}N}{\epsilon}}$. Let $S$ be the set of signals where each $\sigma_i$ is of the form $\frac{1}{2} + \frac{\delta}{M}k_{i}$, for some integer $k_i$ between $-\gamma^{-1} M$ and $\gamma^{-1} M$. Let $\sigma$ be a non-revealer signal sent at node $x$ of protocol $\pi$, and let $p:=\mu_x$ be the belief conditioned on the protocol reaching the node $x$. Then, since $\sigma$ is balanced and has power at least $\delta$, it shifts belief $p$ to $(p-v, p+v)$, for some $v \in \mathbb{R}^{N}$ with $||v|| \geq \delta$. 

Since our signal $\sigma$ has size at most $\gamma^{-1}\delta$, it is contained within a `hypercube' in the space of signals whose vertices belong to $S$. It follows that we can write $\sigma$ as the convex combination $\sum_{k=1}^{2^{N}} w_{k}\sigma^{(k)}$ of $2^{N}$ signals  $\sigma^{(k)}$ in $S$ such that $|\sigma_{i} - \sigma^{(k)}_{i}| \leq \frac{\delta}{M}$ for all $i$ and $k$ (in fact, it can be written as the convex combination of $N$ of these signals, but this does not improve our resulting bound). It follows from Lemma \ref{sigerror} that if $\sigma^{(k)}$ shifts $p$ to $\left(p_{0}^{(k)}, p_{1}^{(k)}\right)$, then $||p_0 - p_{0}^{(k)}|| \leq \frac{9\delta}{M}$ and $||p_1 - p_{1}^{(k)}|| \leq \frac{9\delta}{M}$. 

For ease of notation, let 

\begin{equation*}
E(\sigma) = \left(\sum_{k=1}^{2^{N}}w_{k}C\left(\sigma^{(k)}, p\right)\right) - C(\sigma, p)
\end{equation*}
\noindent
Then, by Lemma \ref{costub}, we have that

\begin{eqnarray*}
E(\sigma) &=& P_0\left(\left(\sum_{k=1}^{2^{N}}w_{k}C\left(p_0^{(k)}\right)\right) - C(p_0)\right)  \\
& & +\; P_1\left(\left(\sum_{k=1}^{2^{N}}w_{k}C\left(p_1^{(k)}\right)\right) - C(p_1)\right)\\
&\leq & \frac{81U_{\gamma}\delta^2 }{M^2}(P_0 + P_1) \\
&=& \frac{81U_{\gamma}\delta^2 }{M^2} 
\end{eqnarray*}

\noindent
(Since $\sigma$ is balanced, it is in fact the case that $P_0=P_1=\frac{1}{2}$, but we only need that $P_0+P_1=1$ above). Now, let $V(p) = ||p||_2^2$, and let $V(\sigma, p) = \frac{1}{2}(V(p-v) + V(p+v)) - V(p)$. Note that

\begin{eqnarray*}
V(\sigma, p) &=& \frac{1}{2}(||p-v||^2 + ||p+v||^2) - ||p||^2\\
&=& ||v||^2\\
&\geq & \delta^2
\end{eqnarray*}

\noindent
It follows that

\begin{equation} \label{errratio}
\frac{E(\sigma)}{V(\sigma, p)} \leq \frac{81U_{\gamma}}{M^2}
\end{equation}

By Lemma \ref{replacesig}, there exists some $k$ such that replacing $\sigma$ with $\sigma^{(k)}$ increases the information cost by at most $q_{\sigma}E(\sigma)$. Repeatedly performing this procedure, we can replace all of the signals in $\pi$ with signals in $S$ while increasing the information cost by at most

\begin{eqnarray*}
\sum_{\sigma \in \pi} q_{\sigma}E(\sigma) &\leq & \frac{81U_{\gamma}}{M^2}\sum_{\sigma \in \pi} q_{\sigma}V(\sigma,p) \\
&=& \epsilon \sum_{\sigma \in \pi}q_{\sigma}V(\sigma, p)
\end{eqnarray*}

\noindent
where the inequality follows by equation \ref{errratio}. Since (by the same logic as that in Lemma \ref{lemsigcost})

\begin{equation*}
\sum_{\sigma \in \pi} q_{\sigma}V(\sigma, p) = \sum_{\mbox{leaf nodes } v}q_v||p_v||^2 \leq 1
\end{equation*}

\noindent
it follows that our new protocol has information cost at most $IC_{\mu}(\pi) + \epsilon$. In addition, since there are only $|\mathcal{A}| + |\mathcal{B}|$ distinct revealer signals, the total number of distinct signals in our new protocol is at most $|S| + (|\mathcal{A}| + |\mathcal{B}|)= (2\gamma^{-1}M)^{N} + (|\mathcal{A}| + |\mathcal{B}|) = Q$, as desired.
\end{proof}

\subsection{Using a bounded number of alternations}\label{sectrounds}

Finally, we show that we can convert a protocol for $f$ that uses a bounded number of distinct signals (yet arbitrarily many of them) into a protocol for $f$ that, while leaking at most $\epsilon$ extra information, uses a bounded number of alternations (steps in the protocol where Alice stops talking and Bob starts talking, or vice versa). 

We achieve this by `bundling' signals of the same type together; that is, at a point in the protocol where Alice would send Bob a certain signal, she may instead send him a bundle of $t$ signals.  Then, the next $t-1$ times Alice would send Bob this signal, Bob instead refers to the next unused signal in the bundle. If there are unused signals in a bundle, this may increase the information cost of the protocol; however, by choosing the size of the bundle cleverly, we can bound the size of this increase.

\begin{definition}
Let $\pi$ be a communication protocol and let $v_1, v_2, \dots, v_k$ be one possible computation path for $\pi$. An \textit{alternation} in this computation path is an index $i$ where the signals at $v_{i}$ and $v_{i+1}$ are sent by different players. The number of alternations in $\pi$ is the maximum number of alternations over all computation paths of $\pi$.
\end{definition}

\begin{theorem} \label{thmaltern}
Let $\pi$ be a communication protocol with information cost $C$ that only uses $Q$ distinct signals. Then, for any $\epsilon > 0$, there exists a communication protocol $\pi'$ that computes the same function as $\pi$ with information cost at most $C+2\epsilon$ but that uses at most

\begin{equation*}
W = \left(\frac{2Q\log N}{\epsilon}+ Q\right)\frac{\log N}{\epsilon}
\end{equation*}
\noindent
alternations. 
\end{theorem}

\begin{proof}
Label our $Q$ different signals $\sigma^{(1)}$ through $\sigma^{(Q)}$. We will reduce the number of alternations in $\pi$ by bundling signals of the same type in large groups. That is, if Alice (at a specific point in the protocol) would send Bob signal $\sigma^{(i)}$, she instead sends Bob $t$ copies of signal $i$ (for an appropriately chosen $t$). Then, the next $t-1$ times in the protocol that Alice would send Bob signal $\sigma^{(i)}$, Bob instead refers to one of the unused $t$ copies Alice originally sent. Once these $t$ copies are depleted and protocol calls for a $(t+1)$st copy, the process repeats and Alice sends a new bundle to Bob (possibly with a different value for $t$).

We choose $t$ as follows. Without loss of generality, assume Alice is sending a bundle of signals $\sigma$ to Bob. Let $\Pi_{pre}$ be the transcript of the protocol thus far. Let $X^{t} = (X_1, X_2, \dots, X_t)$ be a random variable corresponding to $t$ independently generated outputs of $\sigma$. We consider three cases:

\begin{itemize}
\item \textbf{Case 1}: It is the case that 

\begin{equation*}
I(A;X^{1}|\Pi_{pre}B) \geq \frac{\epsilon}{Q}
\end{equation*}

\noindent
In this case we set $t = 1$ (note that this is equivalent to simply following the original protocol).

\item \textbf{Case 2}: There exists a positive $t_{0}$ such that

\begin{equation*}
\frac{\epsilon}{2Q} \leq I(A;X^{t_{0}}|\Pi_{pre}B) \leq \frac{\epsilon}{Q}
\end{equation*}

\noindent
In this case, we set $t = t_{0}$. 

\item \textbf{Case 3}: For all positive $t$,

\begin{equation*}
I(A;X^{t}|\Pi_{pre}B) \leq \frac{\epsilon}{2Q}
\end{equation*}

\noindent
In this case, we set $t$ to be the maximum number of times signal $\sigma$ is ever sent in protocol $\pi$. 
\end{itemize}

The remainder of this proof is divided into three parts. In the first part, we argue that the three cases above are comprehensive. In the second part, we argue that the information cost of this new protocol is at most $C+\epsilon$. Finally, in the third part we argue that this bundling process decreases the total number of alternations to at most $W$.

\paragraph{Cases are comprehensive}

We first argue that the three above cases indeed encompass all possibilities. In particular, the function $I(A;X^{t}|\Pi_{pre}B)$ is non-decreasing in $t$, so it suffices to show that there does not exist a $t$ for which

\begin{equation*}
I(A;X^{t}|\Pi_{pre}B) < \frac{\epsilon}{2Q} ,\;  \frac{\epsilon}{Q} < I(A;X^{t+1}|\Pi_{pre}B)
\end{equation*}

To show this, we claim that $I(A;X^{t+1}|\Pi_{pre}B) - I(A;X^{t}|\Pi_{pre}B)$ is (weakly) decreasing in $t$. This follows from the following chain of inequalities:

\begin{eqnarray}
I(A;X^{t+1}|\Pi_{pre}B) - I(A;X^{t}|\Pi_{pre}B) &= & I(A;X_{t+1}|\Pi_{pre}BX^{t}) \\
&\leq & I(A;X_{t+1}|\Pi_{pre}BX^{t-1}) \label{uselem}\\ 
&=& I(A;X_{t}|\Pi_{pre}BX^{t-1}) \\
&=& I(A;X^{t}|\Pi_{pre}B) - I(A;X^{t-1}|\Pi_{pre}B)
\end{eqnarray}

Here, inequality \ref{uselem} follows from noting that $I(X_{t+1};X_{t}|\Pi_{pre}ABX^{t-1}) = 0$ (indeed, $X_{t+1}$ and $X_{t}$ are conditionally independent given $A$ and $\Pi_{pre}$) and applying Lemma \ref{subadd}.

Now, note that if $I(A;X^{1}|\Pi_{pre}B)$ is greater than $\frac{\epsilon}{2Q}$, we are in either case 1 or case 2. Therefore, assume that $I(A;X^{1}|\Pi_{pre}B) \leq \frac{\epsilon}{2Q}$. Since $I(A;X^{1}|\Pi_{pre}B) - I(A;X^{0}|\Pi_{pre}B) = I(A;X^{1}|\Pi_{pre}B)$, we have that $I(A;X^{t+1}|\Pi_{pre}B) - I(A;X^{t}|\Pi_{pre}B) \leq \frac{\epsilon}{2Q}$. It follows that it is impossible for $I(A;X^{t}|\Pi_{pre}B)$ to be less than $\frac{\epsilon}{2Q}$ while $I(A;X^{t+1}|\Pi_{pre}B)$ is larger than $\frac{\epsilon}{Q}$.

\paragraph{Information leakage is small}

Let $\Pi$ be a random variable corresponding to the transcript of our old protocol, and let $\Pi'$ be a random variable corresponding to the transcript of our new protocol. Note that the only difference between $\Pi'$ and $\Pi$ is that $\Pi'$ contains some excess signals in the form of unfinished bundles. 

For each $i$, let $R_i$ be the random variable corresponding to the excess signals of type $\sigma^{(i)}$ (in particular, $R_i$ is of the form $(X_{u+1}, \dots, X_{t})$ if $u$ out of the $t$ signals in this bundle were used). For $1\leq t \leq Q$, let $R^{t} = (R_1, R_2, \dots, R_t)$. We can then write $\Pi' = \Pi R^{Q}$, from which it follows

\begin{eqnarray*}
I(A;\Pi' |B) &=& I(A;\Pi R^{Q}|B) \\
& = & I(A;\Pi | B) + \sum_{i=1}^{Q} I(A;R_{i}|\Pi BR^{i-1}) \\
& \leq & I(A;\Pi | B) + \sum_{i=1}^{Q} I(A;R_{i}|\Pi B)
\end{eqnarray*}

The last inequality follows from observing that $I(R_{i};R^{i-1}|\Pi AB) = 0$ and applying Lemma \ref{subadd}.

We would now like to show that, for each $i$, $I(A;R_{i}|\Pi B) \leq \frac{\epsilon}{Q}$. To do this, we will define $Y_{i}$ to be the random variable given by

\begin{equation*}
Y_i = (Z_2, \dots, Z_{u}, X_{u+1}, X_{u+2}, \dots, X_{t})
\end{equation*}

Here, $X_{u+1}$ through $X_{t}$ are the elements of $R_i$ (the unused signals in bundle $i$), and $Z_2$ through $Z_u$ are independently sampled Bernoulli random variables with probability $\sigma^{(i)}(A)$ (that is, individually, they are distributed identically to each individual $X_i$ yet independent from $\Pi$ and $X^t$ given $A$). The motivation behind this construction is to pad $R_i$ with additional elements (identically distributed to, yet independent from the $X_i$) as to avoid revealing information about the number $|R_i|$ of unused signals in the bundle (which itself is a random variable that might reveal information about $A$ or $B$). 

Define the random variable $U$ to equal $t-|R_i|$. To begin, note that since the first signal in a bundle is always used, we can always recover $R_i$ given $Y_{i}$ and $U$ (in particular, $R_i$ is a suffix of $Y_i$ of length $t-U$), so $I(A;R_{i}|\Pi B) \leq I(A;UY_{i}|\Pi B) = I(A;Y_{i}|\Pi B) + I(A;U|\Pi B Y_i)$. Since $U$ is recoverable given $\Pi$, $I(A;U|\Pi B Y_i) = 0$, and therefore $I(A;R_{i}|\Pi B) \leq I(A;Y_{i}|\Pi B)$. 

Next, let $\Pi_{pre}$ be the prefix of $\Pi$ up to the point where the last bundle for $\sigma^{(i)}$ was created, and let $\Pi_{fin}$ be the remainder of the transcript (so $\Pi = \Pi_{pre}\Pi_{fin}$). We claim that $I(A;Y_{i}|\Pi B) \leq I(A;Y_{i}|\Pi_{pre} B)$. Again, this follows from observing that $I(Y_{i};\Pi_{fin}|\Pi_{pre}AB) = 0$ and applying Lemma \ref{subadd}. In particular, conditioned on $\Pi_{pre}$, $A$, and $B$, $\Pi_{fin}$ is simply some (randomized) function of $X_1$ through $X_u$, and hence independent from $Y_i$. 

Finally, if this bundle is a Case 1 bundle, then $t=1$ and $Y_{i}$ is empty, so $I(A;Y_{i}|\Pi_{pre}B) = 0$. Otherwise, note that conditioned on $B$ and $\Pi_{pre}$, $Y_i$ is distributed identically to $X^{[2,t]} = (X_2, X_3, \dots, X_t)$ (in particular, again they are both just $t-1$ independent copies of a Bernoulli random variable with probability $\sigma^{(i)}(A)$). Since $X^{t}$ is a superset of $X^{[2,t]}$, it follows that $I(A;Y_{i}|\Pi_{pre}B) = I(A;X^{[2,t]}|\Pi_{pre}B) \leq I(A;X^{t}|\Pi_{pre}B)$ which is at most $\frac{\epsilon}{Q}$.

If $Q_{A}$ of the $Q$ distinct types of signal are sent by Alice, it follows from this argument that $I(A;\Pi' |B) \leq I(A;\Pi|B) + \frac{Q_{A}}{Q}\epsilon$. Since a similar argument establishes that $I(B;\Pi' | A) \leq I(B;\Pi|A) + \frac{Q-Q_{A}}{Q}\epsilon$, it immediately follows that $IC_{\mu}(\pi') \leq IC_{\mu}(\pi) + \epsilon$. 

\paragraph{Number of alternations is small}

Since alternations only occur between bundles, to show that the number of alternations is at most $W$, it suffices to show that the number of bundles sent in an execution of $\pi'$ is at most $W$. To do this, we will modify protocol $\pi'$ by aborting after the $W$th bundle is sent and forcing Alice and Bob to exchange their inputs at this point. We will show that the probability the protocol $\pi'$ uses at least $W$ bundles is at most $\frac{\epsilon}{\log N}$; since the information cost of any protocol is bounded above by $\log N$, this results in an increase of at most $\left(\frac{\epsilon}{\log N}\right)\log N = \epsilon$ in the information cost of our protocol. Combining this with the previous section of the proof, this results in a protocol with information cost at most $IC_{\mu}(\pi) + 2\epsilon$.

Let $M_i$ be the $i$th bundle sent in the protocol, and let $M^{i} = (M_1, M_2, \dots, M_i)$ be the list of the first $i$ bundles sent in the protocol. Note that

\begin{eqnarray}
& & \sum_{i} \left(I(A;M_{i+1}|M^{i}B) + I(B;M_{i+1}|M^{i}A)\right) \\
&=& \sum_{i} \left(I(A;M^{i+1}|B) - I(A;M^{i}|B) + I(B;M^{i+1}|A) - I(B;M^{i}|A)\right)\\
&= & I(A;\Pi'| B) + I(B;\Pi'| A) \\
&= & IC_{\mu}(\pi') \\ 
&\leq & \log N \label{altsumubd}
\end{eqnarray}

Let $p_i$ be the probability that at least $i$ bundles are sent under $\pi$, and let $p_{i,3}$ be the probability that the $i$th bundle sent is a Case 3 bundle. Then, since each non-Case 3 bundle contributes at least $\frac{\epsilon}{2Q}$ to the information cost of the protocol, 

\begin{equation} \label{altsumlbd}
I(A;M_{i+1}|M^{i}B) + I(B;M_{i+1}|M^{i}A) \geq (p_{i} - p_{i,3})\frac{\epsilon}{2Q}
\end{equation}

\noindent
Summing equation \ref{altsumlbd} over all $i$ (and combining with inequality \ref{altsumubd}), we have that

\begin{eqnarray*}
\log N &\geq & I(A;\Pi'|B) + I(B;\Pi'|A) \\
&\geq & \sum_{i} \left(I(A;M_{i+1}|M^{i}B) + I(B;M_{i+1}|M^{i}A)\right) \\
&\geq & \frac{\epsilon}{2Q}\left(\sum_{i}p_i - \sum_{i}p_{i,3}\right)
\end{eqnarray*}

Note that $\sum_{i}p_{i,3}$ is the expected number of Case 3 bundles sent. Since we send at most one Case 3 bundle of each type, this sum is at most $Q$. It follows that

\begin{equation*}
\sum_{i}p_{i} \leq \frac{2Q\log N}{\epsilon} + Q
\end{equation*}

Finally, since the $p_{i}$ are non-increasing, the probability we send at least $W$ bundles is at most

\begin{equation*}
p_{W} \leq \frac{1}{W}\sum_{i}p_{i} \leq \frac{\epsilon}{\log N}
\end{equation*}

\noindent
as desired.

\end{proof}

\subsection{Computing Information Complexity}\label{sectic}

Combining the results of Theorems \ref{perturbmu}, \ref{thmboundary}, \ref{thmbounded}, \ref{thmsignals}, and \ref{thmaltern}, we obtain the following result.

\begin{theorem}\label{convthm}
Let $\pi$ be a communication protocol with information cost $C$ that successfully computes function $f$ over inputs drawn from distribution $\mu$ over $\mathcal{A}\times\mathcal{B}$. Then there exists a protocol $\pi'$ with information cost at most $C+\epsilon$ that also successfully computes $f$ over inputs drawn from $\mu$, but that uses at most $w(f, \epsilon)$ alternations where

\begin{equation}
w(f, \epsilon) = (N\epsilon^{-1})^{O(N)}
\end{equation}

\noindent
where $N = |\mathcal{A}\times\mathcal{B}|$.
\end{theorem}
\begin{definition}
Let the $W$-alternation information cost of $f$, $IC_{W,\mu}(f)$, equal $\inf_{\pi} IC_{\mu}(\pi)$, where the infimum is taken over all protocols $\pi$ that successfully compute $f$ that use at most $W$ alternations. 
\end{definition}
\begin{corollary}
We have that

\begin{equation*}
IC_{\mu}(f) \leq IC_{w(f, \epsilon),\mu}(f) \leq IC_{\mu}(f)+\epsilon
\end{equation*}
\end{corollary}

The following result of Braverman and Rao provides a link between the communication complexity and information complexity of protocols restricted to at most $W$ alternations. 

\begin{lemma}\label{brlemma}
Let $\pi$ be a protocol with information cost $I$ that uses at most $W$ alternations. Then, for every $\epsilon > 0$, there exists a protocol $\tau$ such that 

\begin{enumerate}[i)]
\item with probability at least $1-\epsilon$, at the end of protocol $\tau$, Alice and Bob output a valid transcript for $\pi$ (distributed identically to $\pi(A,B)$). 
\item the communication cost of $\tau$ is at most $I + O(\sqrt{WI} + W) + 2W\log(W/\epsilon)$. 
\end{enumerate}
\end{lemma}
\begin{proof}
See Corollary 2.2 of \cite{BR11}.
\end{proof}

Define $IC_{\mu}(f,\epsilon)$ to be equal to $\inf_{\pi} IC_{\mu}(\pi)$, where the infimum is taken over all protocols $\pi$ that successfully compute $f$ with probability at least $1-\epsilon$. The following lemma relates $IC_{\mu}(f,\epsilon)$ to $IC_{\mu}(f)$.

\begin{lemma}\label{contlemma}
For all $\epsilon \in (0, \rho^8)$ (where as before, $\rho = \min_{a,b} \mu(a,b)$),

\begin{equation*}
IC_{\mu}(f) \leq IC_{\mu}(f, \epsilon) + 2\left(h\left(1 - \frac{2N\epsilon^{1/4}}{\rho}\right) + 2(\log N + 2)\frac{N\epsilon^{1/4}}{\rho}\right)
\end{equation*}
\end{lemma}
\begin{proof}
See Lemma 6.3 in \cite{BGPW}.
\end{proof}

Likewise, define $CC_{\mu}(f, \epsilon)$ to be equal to $\inf_{\pi} CC(\pi)$, where the infimum is taken over all protocols $\pi$ that successfully compute $f$ with probability at least $1-\epsilon$ (when inputs are drawn from distribution $\mu$). The following theorem relates $IC_{\mu}(f, \epsilon)$ to $CC_{\mu}(f, \epsilon)$. 

\begin{theorem}\label{boundcc}
Let 
\begin{equation*}
L_{\mu}(f, \epsilon) = 2\left(h\left(1 - \frac{2N\epsilon^{1/4}}{\rho}\right) + 2(\log N + 2)\frac{N\epsilon^{1/4}}{\rho}\right)
\end{equation*}

\noindent
and

\begin{equation*}
U_{\mu}(f, \epsilon) = \epsilon + O\left(\sqrt{w(f, \epsilon)(IC_{\mu}(f)+\epsilon)} + w(f, \epsilon)\right) + 2w(f, \epsilon)\log\left(\frac{w(f, \epsilon)}{\epsilon}\right)
\end{equation*}

\noindent
We have that

\begin{equation*}
IC_{\mu}(f) - L_{\mu}(f, \epsilon) \leq CC_{\mu}(f, \epsilon) \leq IC_{\mu}(f) + U_{\mu}(f, \epsilon)
\end{equation*}
\end{theorem}
\begin{proof}
To prove the first two inequalities, first note that $IC_{\mu}(f,\epsilon) \leq CC_{\mu}(f,\epsilon)$; this follows from the fact that for any protocol $\pi$, $IC_{\mu}(\pi) \leq CC(\pi)$. Then, by Lemma \ref{contlemma}, $IC_{\mu}(f) \leq IC_{\mu}(f,\epsilon) + L_{\mu}(f, \epsilon)$, from which the leftmost inequality follows. 

To prove the right inequality, set $W = w(f, \epsilon)$ in Lemma \ref{brlemma}, and apply the fact that $IC_{w(f, \epsilon), \mu}(f) \leq IC_{\mu}(f) + \epsilon$.
\end{proof}

Theorem \ref{boundcc} shows that if we can compute the $\epsilon$-error communication complexity of $f$, we can approximate the information complexity of $f$ to within an additive factor of $\max(L_{\mu}(f, \epsilon), U_{\mu}(f, \epsilon))$. Unfortunately, while we can make $L_{\mu}(f, \epsilon)$ arbitrarily small by decreasing $\epsilon$, $U_{\mu}(f, \epsilon)$ may be large. To remedy this, we apply the following direct sum results.

\begin{lemma}\label{directsum}
We have that

\begin{equation*}
IC_{\mu}(f^n, \epsilon) = nIC_{\mu}(f, \epsilon) 
\end{equation*}

\noindent
In particular, for $\epsilon=0$,

\begin{equation*}
IC_{\mu}(f^n) = n IC_{\mu}(f) 
\end{equation*}
\end{lemma}
\begin{proof}
See Theorem 4.3 in \cite{Br12}.
\end{proof}

\begin{lemma}\label{directsumalt}
We have that

\begin{equation*}
\frac{IC_{w(f, \epsilon),\mu}(f^n)}{n} \leq IC_{w(f, \epsilon), \mu}(f)
\end{equation*}

\end{lemma}

\begin{proof}
Let $\pi$ be the protocol that computes $f$ which requires $w(f,\epsilon)$ rounds and has information cost $IC_{w(f,\epsilon), \mu}(f)$. By running $n$ copies of $\pi$ in parallel, we can construct a protocol for $f^n$ which still requires only $w(f, \epsilon)$ rounds and has information cost $nIC_{w(f, \epsilon), \mu}$. The result follows.
\end{proof}

\begin{corollary}\label{unmubound}
We have that

\begin{equation}
CC_{\mu}(f^n, \epsilon) \leq IC_{\mu}(f^n) + U_{n, \mu}(f, \epsilon)
\end{equation}

\noindent 
where

\begin{equation}
U_{n, \mu}(f, \epsilon) = n\epsilon + O\left(\sqrt{n \cdot w(f, \epsilon)(IC_{\mu}(f)+\epsilon)} + w(f, \epsilon)\right) + 2w(f, \epsilon)\log\left(\frac{w(f, \epsilon)}{\epsilon}\right)
\end{equation}

\end{corollary}
\begin{proof}
We follow the proof of the upper bound in Theorem \ref{boundcc}, but instead of setting $W = w(f^n, \epsilon)$ in Lemma \ref{brlemma}, we set $W = w(f, \epsilon)$. Then, by Lemma \ref{directsumalt}, $IC_{w(f,\epsilon),\mu}(f^n) \leq nIC_{w(f,\epsilon), \mu}(f) \leq n(IC_{\mu}(f) + \epsilon)$.  Applying this fact, we obtain our desired result. 
\end{proof}

\begin{corollary}\label{dsumcc}
For any $n \geq 1$,
\begin{equation*}
IC_{\mu}(f) - L_{\mu}(f, \epsilon) \leq \frac{CC_{\mu}(f^n, \epsilon)}{n} \leq IC_{\mu}(f) + \frac{U_{n, \mu}(f, \epsilon)}{n}
\end{equation*}
\end{corollary}
\begin{proof}
To prove the lower bound, recall that $IC_{\mu}(f^n, \epsilon) \leq CC_{\mu}(f^n, \epsilon)$. Dividing by $n$ and applying the result of Lemma \ref{directsum}, we obtain that $IC_{\mu}(f, \epsilon) \leq \frac{CC_{\mu}(f^n, \epsilon)}{n}$. Applying the lower bound from Theorem \ref{boundcc}, we have that $IC_{\mu}(f) - L_{\mu}(f, \epsilon) \leq IC_{\mu}(f, \epsilon)$, and hence that

\begin{equation}
IC_{\mu}(f) - L_{\mu}(f, \epsilon) \leq \frac{CC_{\mu}(f^n, \epsilon)}{n}
\end{equation}

To show the upper bound, simply divide both sides of Corollary \ref{unmubound} by $n$ and apply Lemma \ref{directsum}.
\end{proof}

We can now proceed to prove our main theorem.

\begin{proof}[Proof of Theorem \ref{mainthm}]
Fix an $\alpha > 0$; we will show how to approximate the information complexity of $f$ to within an additive factor of $\alpha$. First, note that since $L_{\mu}(f, \epsilon)$ is decreasing in $\epsilon$, we can choose $\epsilon$ small enough so that both $L_{\mu}(f,\epsilon) \leq \alpha$ and $\epsilon \leq \frac{\alpha}{2}$. Secondly, note that 

\begin{equation}
U_{n,\mu}(f, \epsilon) \leq n\epsilon + \sqrt{n}U_{\mu}(f, \epsilon)
\end{equation}

\noindent
It follows from Corollary \ref{dsumcc} that 

\begin{equation}
\frac{CC_{\mu}(f^n, \epsilon)}{n} \leq IC_{\mu}(f) + \epsilon + \frac{U_{\mu}(f, \epsilon)}{\sqrt{n}}
\end{equation}

If we choose $n$ large enough so that $\frac{U_{\mu}(f, \epsilon)}{\sqrt{n}} \geq \frac{\alpha}{2}$, then it follows from Corollary \ref{dsumcc}, that $CC_{\mu}(f^{n},\epsilon)/n$ approximates $IC_{\mu}(f)$ to within an additive factor of $\alpha$. 

Note that $n\log N$ is an upper bound on $CC_{\mu}(f^{n}, \epsilon)$. We can therefore compute $CC_{\mu}(f^{n}, \epsilon)$ simply by enumerating all protocols of depth at most $n\log N$, checking which protocols compute $f$ successfully at least $(1-\epsilon)$ proportion of the time, and taking the minimal communication cost of such protocols. This completes the proof that information complexity is computable.

To obtain explicit asymptotic bounds on $n$, note that to ensure $L_{\mu}(f,\epsilon) \leq \alpha/2$, it suffices to take $\epsilon = \tilde{O}\left(\alpha^8N^{-4}\rho^{4}\right) = \tilde{O}\left(\alpha^{16}N^{-8}\right)$ (by the proof of Theorem \ref{perturbmu}, we know that we can ensure $\rho = \tilde{\Omega}(\alpha^2N^{-1})$). For this value of $\epsilon$, in order to choose $n$ so that $\frac{U_{\mu}(f^n,\epsilon)}{n} \leq \alpha/2$, it suffices to take 

\begin{equation*}
n = O(w(\epsilon)\log w(\epsilon)/\alpha^2) = (N\alpha^{-1})^{O(N)}
\end{equation*}

\noindent
and hence it suffices to enumerate protocols with up to a maximum depth $d$ on the order of $(N\alpha^{-1})^{O(N)}$. The number of such protocols is at most

\begin{equation*}
2^{N2^{d}} = 2^{\exp\left((N\alpha^{-1})^{O(N)}\right)}
\end{equation*}

Since each protocol with depth $d$ can be checked for correctness in time $O(Nd)$ (by checking all possible $N$ pairs of inputs), this is also a bound on the time complexity of this algorithm, as desired.
\end{proof}

\begin{remark}\label{extremark}
The techniques in this section do not immediately extend to the case of external information complexity (in particular, no direct sum statement analogous to Lemma \ref{directsum} is known for external information complexity). Instead, to prove the analogue of Theorem \ref{mainthm} for external information complexity, we can proceed from Theorem \ref{convthm} by applying the results of Ma and Ishtar to approximate arbitrarily closely the $W$-alternation external information complexity of $f$ (see section II.B of \cite{MI09}). 
\end{remark}

\subsection*{Acknowledgments}

We would like to thank Ankit Garg and Noga Ron-Zewi for providing helpful comments on an earlier draft of this paper.

\end{document}